\documentclass{article}
\usepackage{ijcai21}

\usepackage{times}
\usepackage{soul}
\usepackage{url}
\usepackage[hidelinks]{hyperref}
\usepackage[utf8]{inputenc}
\usepackage[small]{caption}
\usepackage{graphicx}
\usepackage{amsmath}
\usepackage{amsthm}
\usepackage{booktabs}
\usepackage{algorithm}
\usepackage{algorithmic}
\newtheorem{theorem}{Theorem}

\title{Online Risk-Averse Submodular Maximization}

\author{
Tasuku Soma$^1$$^2$
\and
Yuichi Yoshida$^3$
\affiliations
$^1$ The University of Tokyo
$^2$ Massachusetts Institute of Technology\\
$^3$ National Institute of Informatics\\
\emails
tasuku\_soma@mist.i.u-tokyo.ac.jp,
tasuku@mit.edu,
yyoshida@nii.ac.jp
}

\usepackage{amsmath,amsthm,amssymb}
\usepackage{subfig,booktabs}
\newcommand{\R}{\mathbb{R}}

\newcommand{\caB}{\mathcal{B}}
\newcommand{\caD}{\mathcal{D}}
\newcommand{\caP}{\mathcal{P}}
\newcommand{\caI}{\mathcal{I}}
\newcommand{\caZ}{\mathcal{Z}}
\newcommand{\eps}{\varepsilon}
\newcommand{\bx}{{\boldsymbol{x}}}
\newcommand{\by}{{\boldsymbol{y}}}
\newcommand{\bd}{{\boldsymbol{d}}}
\newcommand{\bg}{{\boldsymbol{g}}}
\newcommand{\br}{{\boldsymbol{r}}}

\newcommand{\bz}{{\boldsymbol{z}}}
\newcommand{\bfzero}{\mathbf{0}}

\newcommand{\regFPL}{r_{\mathrm{FPL},s}}
\newcommand{\regOGD}{r_{\mathrm{OGD}}}
\DeclareMathOperator*{\E}{\mathbf{E}}
\DeclareMathOperator{\VaR}{VaR}
\DeclareMathOperator{\CVaR}{CVaR}

\DeclareMathOperator{\regret}{regret}
\DeclareMathOperator{\proj}{proj}

\DeclareMathOperator*{\argmax}{argmax}

\newtheorem{lemma}[theorem]{Lemma}

\theoremstyle{definition}

\newtheorem{assumption}[theorem]{Assumption}
\newtheorem{remark}[theorem]{Remark}

\usepackage{mathtools}
\DeclarePairedDelimiter{\norm}{\lVert}{\rVert}
\DeclarePairedDelimiter{\abs}{\lvert}{\rvert}
\DeclarePairedDelimiter{\inprod}{\langle}{\rangle}

\usepackage{xcolor}

\usepackage{natbib}
\usepackage{multibib}

\begin{document}

\maketitle

\begin{abstract}
    We present a polynomial-time online algorithm for maximizing the conditional value at risk (CVaR) of a monotone stochastic submodular function.
    Given $T$ i.i.d.\ samples from an underlying distribution arriving online, our algorithm produces a sequence of solutions that converges to a ($1-1/e$)-approximate solution with a convergence rate of $O(T^{-1/4})$ for monotone continuous DR-submodular functions.
    Compared with previous offline algorithms, which require $\Omega(T)$ space, our online algorithm only requires $O(\sqrt{T})$ space.
    We extend our online algorithm to portfolio optimization for monotone submodular set functions under a matroid constraint.
    Experiments conducted on real-world datasets demonstrate that our algorithm can rapidly achieve CVaRs that are comparable to those obtained by existing offline algorithms.
\end{abstract}

\allowdisplaybreaks

\section{Introduction}
\emph{Submodular function maximization} is a ubiquitous problem naturally arising in broad areas such as machine learning, social network analysis, economics, combinatorial optimization, and decision-making~\citep{Krause2014survey,Buchbinder2018}.
Submodular maximization in these applications is \emph{stochastic} in nature, i.e., the input data could be a sequence of samples drawn from some underlying distribution, or there could be some uncertainty in the environment.
In this paper, we consider nonnegative monotone \emph{continuous DR-submodular}\footnote{See Section~\ref{sec:pre} for the definition of continuous DR-submodularity.}~\citep{Bian2017} functions $F(\bx; z)$ parameterized by a random variable $z$ drawn from some distribution $\caD$.
The simplest approach for such stochastic submodular objectives is to maximize the expectation $\E_{z \sim \caD}[F(\bx; z)]$, which has been extensively studied~\citep{Karimi2017a,Hassani2017a,Mokhtari2018,Karbasi2019}.

However, in real-world decision-making tasks in finance, robotics, and medicine, we sometimes must be \emph{risk-averse}:
We want to minimize the risk of suffering a considerably small gain rather than simply maximizing the expected gain~\citep{Mansini2007,Yau2011,Tamar2015}.
In medical applications, for example, we must avoid catastrophic events such as patient fatalities.
In finance and robotics, all progress ceases when poor decisions cause bankruptcies or irreversible damage to robots.

\emph{Conditional value at risk (CVaR)} is a popular objective for such risk-averse domains~\citep{Rockafellar2000,Krokhmal2002}.
Formally, given a parameter $\alpha \in [0,1]$, the CVaR of a feasible solution $\bx \in \R^n$ is defined as
\[
     \CVaR_{\alpha, \caD}(\bx) = \E_{z\sim\caD}[F(\bx;z) \mid F(\bx; z) \leq \VaR_{\alpha,\caD}(\bx)],
\]
where $\VaR_{\alpha,\caD}(X)$ is the $\alpha$-quantile of the random variable $F(\bx;z)$, i.e.,
\[
    \VaR_{\alpha,\caD}(X) = \sup\left\{\tau \in \R: \Pr_{z\sim\caD}(F(\bx;z) \leq \tau) \leq \alpha \right\} .
\]
Note that when $\alpha = 1$, the CVaR is simply the expected value of $F(\bx;z)$.
Since $\alpha$ is typically set to be $0.10$ or $0.05$ in practice, we assume that $\alpha$ is a fixed constant throughout this paper.
When $\alpha$ is clear from the context, we omit it from the notations.
CVaR can also be characterized by a variational formula:
\[
    \CVaR_\caD(\bx) = \max_{\tau\in[0,1]} \tau - \frac{1}{\alpha}\E_{z\sim\caD}{[\tau - F(\bx;z)]}_+,
\]
where ${[\cdot]}_+ = \max\{\cdot, 0\}$ \citep{Rockafellar2000}.

\citet{Maehara2015} initiated the study of maximizing the CVaR of stochastic submodular \emph{set} functions.
It was shown that the CVaR of a stochastic submodular set function is not necessarily submodular, and that it is impossible to compute a single set that attains any multiplicative approximation to the optimal CVaR.
\citet{Ohsaka2017} introduced a relaxed problem of finding a portfolio over sets rather than a single set, and devised the first CVaR maximization algorithm with an approximation guarantee for the influence maximization problem~\citep{Kempe2003}, a prominent example of discrete submodular maximization.
\citet{Wilder2018} further considered this approach and devised an algorithm called \textsc{RASCAL} for maximizing the CVaR of continuous DR-submodular functions subject to a down-closed convex set.

The algorithms mentioned above are \emph{offline (batch) methods}, i.e., a set of samples drawn from the underlying distribution $\caD$ is given as the input.
However, because the number of samples needed to accurately represent $\caD$ can be exceedingly large, it is often inefficient or even impossible to store all the samples in memory.
Further, when a new sample is observed, these offline algorithms must be rerun from scratch.

In the machine learning community, \emph{online methods}, which can efficiently handle large volumes of data, have been extensively studied~\citep{Hazan2016OCO}.
Online methods read data in a streaming fashion, and update the current solution using a few data elements stored in memory.
Further, online methods can update a solution at a low computational cost.

\subsection{Our contributions}
In this work, we propose \emph{online algorithms} for maximizing the CVaR of stochastic submodular objectives in continuous and discrete settings.

\paragraph{Continuous setting} Let $z_1, \dots, z_T$ be i.i.d.\ samples drawn from $\caD$ arriving sequentially and $K \subseteq \R^n$ be a down-closed convex set.
Our main result is a polynomial-time online algorithm, \textsc{StochasticRASCAL}, that finds $\bx \in K$ such that
\begin{align*}
    \E[\CVaR_\caD(\bx)] \geq \left(1-\frac{1}{e}\right) \CVaR_\caD(\bx^*) - O(T^{-1/4}),
\end{align*}
for any $\bx^* \in K$, where the expectation is taken over $z_1, \dots, z_T$ and the randomness of the algorithm.
\textsc{StochasticRASCAL} only stores $O(\sqrt{T})$ data in memory, drawing a contrast to \textsc{RASCAL}, which store all $T$ data points.
Note that the approximation ratio of $1-1/e$ is optimal for any algorithm that performs polynomially many function value queries even if $\alpha=1$~\citep{Vondrak2013}.
We also conduct several experiments on real-world datasets to show the practical efficiency of our algorithm.
We demonstrate that our algorithm rapidly achieves CVaR comparable to that obtained by known offline methods.

\paragraph{Discrete setting}
As an application of the above algorithm, we devise an online algorithm to create a portfolio that maximizes the CVaR of a discrete submodular function subject to a matroid constraint.
Let $f(X; z) : 2^V \to [0,1]$ be a monotone stochastic submodular set function on a ground set $V$ and $\caI \subseteq 2^V$ be a matroid.
The goal is to find a portfolio $\caP$ of feasible sets in $\caI$ that maximizes
\[
    \CVaR_\caD(\caP) = \max_{\tau\in[0,1]} \tau - \frac{1}{\alpha}\E_{z\sim\caD}[\tau - \E_{X\sim\caP}f(X;z)]_+,
\]
given i.i.d.\ samples from $\caD$.
We show that this problem can be reduced to online CVaR maximization of continuous DR-submodular functions on a matroid polytope, and devise a polynomial-time online approximation algorithm for creating a portfolio $\caP$ such that
\[
    \E[\CVaR_\caD(\caP)] \geq \left(1-\frac{1}{e}\right) \CVaR_\caD(\caP^*) - O(T^{-1/10})
\]
for any portfolio $\caP^*$ over feasible sets.
Note that this algorithm is the first online algorithm that converges to a $(1-1/e)$-approximation portfolio in the discrete setting, which generalizes the known offline algorithm~\citep{Wilder2018} to the i.i.d.~setting.

\subsection{Our techniques}
To analyze our online algorithms, we introduce a novel adversarial online learning problem, which we call \emph{adversarial online submodular CVaR learning}.
This online learning problem is described as follows.
For $t = 1, \dots, T$, the learner chooses $\bx_t \in K$ and $\tau_t \in [0,1]$ possibly in a randomized manner.
After $\bx_t$ and $\tau_t$ are chosen, the adversary reveals a monotone continuous DR-submodular function $F_t : K \to [0,1]$ to the learner.
The goal of the learner is to minimize the approximate regret
\[
    \regret_{1-1/e}(T) = \left(1-\frac{1}{e}\right)\sum_{t=1}^T H_t(\bx^*, \tau^*) - \sum_{t=1}^T H_t(\bx_t, \tau_t)
\]
for arbitrary $\bx^* \in K$ and $\tau^* \in [0,1]$,
where the function $H_t$ is given by
\[
    H_t(\bx, \tau) = \tau - \frac{1}{\alpha}{[\tau - F_t(\bx)]}_+.
\]
We devise an efficient algorithm that achieves $O(T^{3/4})$ approximate regret in expectation.
Further, we show that, given an online algorithm with a sublinear approximate regret, we can construct an online algorithm that achieves a $(1-1/e)$-approximation to CVaR maximization, whose convergence rate is $\E[\regret_{1-1/e}(T)]/T$.
Combining these results, we obtain an online $(1-1/e)$-approximation algorithm for CVaR maximization with a convergence rate of $O(T^{-1/4})$.

We remark that adversarial online submodular CVaR learning may be of interest in its own right:
Although the objective function $H_t$ is neither monotone nor continuous DR-submodular in general, we can design an online algorithm with a sublinear ($1-1/e$)-regret by exploiting the underlying structure of $H_t$.
As per our knowledge, an online algorithm for non-monotone and non-DR-submodular maximization does not exist in the literature.

\subsection{Related work}

Several studies focused on CVaR optimization in the adversarial online settings and i.i.d.\ settings.
\citet{Tamar2015} studied CVaR optimization over i.i.d.\ samples and analyzed stochastic gradient descent under the strong assumption that CVaR is continuously differentiable.
Recently, \citet{Cardoso19a} introduced the concept of the CVaR regret for convex loss functions and provided online algorithms for minimizing the CVaR regret under bandit feedback.

Online and stochastic optimization of submodular maximization have been extensively studied in \citet{Streeter2008,Streeter2009,Golovin2014,Karimi2017a,Hassani2017a,Mokhtari2018,Chen2018,Roughgarden2018,Soma2019,Karbasi2019,Zhang2019}.
These studies optimize either the approximate regret or the expectation and do not consider CVaR.

Another line of related work is \emph{robust submodular maximization}~\citep{Krause2008b,Chen2017a,Anari2019}.
In robust submodular maximization, we maximize the minimum of $N$ submodular functions, i.e.,  $\min_{i=1}^N f_i(X)$.
Robust submodular maximization is the limit of CVaR maximization, where $\caD$ is the uniform distribution over $N$ values and $\alpha \to 0$.
Recently, \citet{Staib2019b} proposed \emph{distributionally robust submodular optimization}, which maximizes $\min_{\caD \in P} \E_{z \sim \caD} f(X; z)$ for an uncertainty set $P$ of distributions.
It is known that CVaR can be formulated in the distributionally robust framework~\citep{Shapiro2014}.
However, the algorithms proposed by \citet{Staib2019b} require that $P$ is a subset of the $N$-dimensional probability simplex;
moreover, their time complexity depends on $N$.
Our algorithms work even if $N$ is infinite.

\subsection{Organization of this paper}
This paper is organized as follows.
Section~\ref{sec:pre} introduces the background of submodular optimization.
Sections~\ref{sec:DR} and~\ref{sec:matroid} describe our algorithms for continuous and discrete setting, respectively.
Section~\ref{sec:experiments} present experimental results using real-world dataset.
The omitted analysis and the details of adversarial setting can be found in Appendix.

\section{Preliminaries}\label{sec:pre}
Throughout the paper, $V$ denotes the ground set and $n$ denotes the size of the ground set.
For a set function $f:2^V \to \R$, the \emph{multilinear extension} $F : [0,1]^V \to \R$ is defined as $F(\bx) = \sum_{S \subseteq V}f(S)\prod_{i \in S}x_i \prod_{i \notin S}(1-x_i)$.
For a matroid on $V$, the \emph{base polytope} is the convex hull of bases of the matroid.
It is well-known that the linear optimization on a base polytope can be solved by the greedy algorithm~\citep{Fujishige2005}.

We denote the Euclidean norm and inner product by $\norm{\cdot}$ and $\inprod{\cdot, \cdot}$, respectively.
The $\ell^p$ norm ($1\leq p \leq \infty$) is denoted by $\norm{\cdot}_p$.
The Euclidean projection of $\bx$ onto a set $K$ is denoted by $\proj_K(\bx)$.
A convex set $K \subseteq \R_{\geq 0}^n$ is said to be \emph{down-closed} if $\by \in K$ and $\bfzero \leq \bx \leq \by$ imply $\bx \in K$.
A function $f: \R^n \to \R$ is said to be \emph{$L$-Lipschitz (continuous)} for $L > 0$ if $\abs{f(\bx)-f(\by)} \leq L\norm{\bx-\by}$ for all $\bx, \by$.
We say that $f$ is \emph{$\beta$-smooth} for $\beta > 0$ if $f$ is continuously differentiable and $\norm{\nabla f(\bx) - \nabla f(\by)} \leq \beta \norm{\bx - \by}$.
A smooth function $F: \R^n \to \R$ is said to be \emph{continuous DR-submodular}~\citep{Bian2017} if $\frac{\partial^2 F}{\partial x_i\partial x_j} \leq 0$ for all $i, j$.
The multilinear extension of a submodular function is known to be DR-submodular~\citep{Calinescu2011}.
The continuous DR-submodularity implies \emph{up-concavity}:
For a continuous DR-submodular function $F$, $\bx \in \R^n$, and $\bd \geq \bfzero$, the univariate function $t \mapsto F(\bx + t\bd)$ is concave.

The uniform distribution of a set $K$ is denoted by $\mathrm{Unif}(K)$.
The standard normal distribution is denoted by $N(\bfzero, I)$.

\section{CVaR Maximization of Continuous DR-submodular Functions}\label{sec:DR}
We present our online algorithm for CVaR maximization via i.i.d.\ samples.
Let $F_t : \R^n \to [0,1]$ be the monotone continuous DR-submodular function corresponding to the $t$-th sample $z_t$, i.e.,
\[
    F_t(\bx) = F(\bx; z_t)
\]
for $t=1, \dots, T$.
Similarly, define an auxiliary function $H_t$ with respect to $z_t$ by
\[
    H_t(\bx, \tau) = \tau - \frac{1}{\alpha}{[\tau - F_t(\bx)]}_+
    = \tau - \frac{1}{\alpha}{[\tau - F(\bx; z_t)]}_+.
\]
for $t=1, \dots, T$.
Let $K \subseteq \R_{\geq 0}^n$ be a down-closed convex set.
Formally, we make the following very mild assumptions on $F_t$ and $K$.

\begin{assumption}\label{asmp:DR}
    \phantom{x}
    \begin{itemize}
        \item For all $t$, $F_t$ is $L$-Lipschitz and $\beta$-smooth, and $\norm{\nabla F_t} \leq G$.
        \item The diameter of $K$ is bounded by $D$.
        \item We are given a linear optimization oracle over $K$.
    \end{itemize}
\end{assumption}
When the underlying norm is the $\ell^p$-norm with $p \neq 2$, we write $G_p$ to emphasize it.
For example, if $F_t$ is the multilinear extension of a submodular set function $f_t:2^V \to [0,1]$ and $K$ is the base polytope of a rank-$k$ matroid, we have $L = \beta = O(n)$, $G_\infty = 1$, $D=O(\sqrt{k})$.

Our algorithm borrows some ideas from an algorithm called \textsc{RASCAL}~\citep{Wilder2018}.
First, we define a smoothed auxiliary function $\tilde{H}_t: {[0,1]}^V \times [0,1] \to [0,1]$ as
\[
    \tilde{H}_t(\bx, \tau) = \frac{1}{u} \int_0^u \left( \tau + \xi - \frac{1}{\alpha}{[\tau + \xi - F_t(\bx)]}_+ \right) d\xi,
\]
where $u > 0$ is a smoothing parameter specified later.
This smoothing guarantees that $\tilde{H}_t$ is differentiable for all $\bx$ and has Lipschitz continuous gradients.

\begin{lemma}[Lemma~6 of~\cite{Wilder2018}]
    \phantom{x}
    \begin{enumerate}
        \item $\abs{H_t(\bx, \tau) - \tilde{H}_t(\bx, \tau)} \leq \frac{u(1+1/\alpha)}{2}$ for all $\bx$ and $\tau$.
        \item If $F_t$ is $L$-Lipschitz and $\beta$-smooth, and $\norm{\nabla F_t} \leq G$, then $\nabla_\bx \tilde{H}_t$ is $\frac{1}{\alpha}(\beta+\frac{LG}{u})$-Lipschitz.
    \end{enumerate}
\end{lemma}

\begin{lemma}[{\citet{Wilder2018}}]
    The function $\max_{\tau} \tilde{H}_t(\cdot, \tau)$ is monotone and up-concave.
\end{lemma}

\subsection{StochasticRASCAL}
We now formally describe our algorithm, \textsc{StochasticRASCAL}.
We note that \textsc{RASCAL} runs the Frank-Wolfe algorithm on a function $\max_{\tau} \sum_{t=1}^T \tilde{H}_t(\cdot, \tau)$.
Owing to the up-concavity and smoothness properties of this function, one can obtain $(1-1/e)$-approximation.
However, in our online setting, we cannot evaluate this function because $\tilde{H}_t$ will be revealed online, and hence we cannot simply run \textsc{RASCAL}.

To overcome the issue above, first we split the $T$ samples into \emph{mini-batches} of length $B$, which we specify later.
The key idea is to use the following objective function
\[
    \bar{H}_b(\bx) = \max_{\tau} \frac{1}{B} \sum_{t=(b-1)B + 1}^{bB} \tilde{H}_t(\bx, \tau)
\]
for each $b$-th mini-batch ($b = 1, \dots, T/B$).
We can see that $\bar{H}_b(\bx)$ is monotone, up-concave, and can be evaluated only using samples in the $b$-th mini-batch.
Then, we run a perturbed version of Frank-Wolfe algorithm~\citep{Golovin2014,Bian2017} on $\bar{H}_b$.
More formally, we first initialize $\bx_{b+1}^0 = \bfzero$ and for each $s = 0, \delta, 2\delta, \dots, 1-\delta$, we perform the update
\[
    \bx_{b}^{s + \delta} \gets \bx_{b}^s + \delta \bd_{b}^s
\]
where $\delta$ is the step size, $\bd_{b}^s$ is a solution to a perturbed linear optimization problem on $K$:
\[
    \bd_{b}^s \in \argmax_{\bd \in K} \inprod*{\lambda \sum_{b'=1}^{b} {\nabla \bar{H}_{b'}}(\bx_{b'}^s) + \br_{b}^s, \bd}.
 \]
Here, $\br_b^s \sim \caD_\mathrm{FPL}$ is a perturbation vector.
This perturbation trick aims to stabilize the algorithm so that we can maximize the true objective $\max_{\tau}\sum_{t=1}^T \tilde{H}_t(\cdot, \tau)$ using only mini-batch objectives.

In each iteration of continuous greedy, we need the gradients $\nabla \bar{H}_{b}(\bx_{b}^s)$, which in turn requires us to compute $\argmax_{\tau} \frac{1}{B}\sum_{t=(b-1)B + 1}^{bB} \tilde{H}_t(\bx_{b}^s, \tau)$.
These gradients and the optimal $\tau$ can be computed by \textsc{SmoothGrad} and \textsc{SmoothTau} subroutines, respectively, which were proposed in \citet{Wilder2018};
See Algorithms~\ref{alg:smooth-grad} and~\ref{alg:smooth-tau}.

Let us write $\bx_b := \bx_b^1$ for $b = 1, \dots, T/B$.
The final output of \textsc{StochasticRASCAL} is $\bx_{b'}$ for a random index $b'$ chosen uniformly at random from $\{1, 2, \dots, T/B\}$.
The pseudocode of \textsc{StochasticRASCAL} is presented in Algorithm~\ref{alg:stochastic-RASCAL}.

\begin{algorithm}[t]
    \caption{\textsc{StochasticRASCAL}}\label{alg:stochastic-RASCAL}
\begin{algorithmic}[1]
    \REQUIRE{learning rates $\lambda > 0$, step size $\delta > 0$, perturbation distribution $\caD_\mathrm{FPL}$, smoothing parameter $u>0$, and mini-batch size $B > 0$}
    \STATE Initialize $\bx_1 \in K$ arbitrary.
    \FOR{$b=1,\dots,T/B$}
    \STATE Observe samples $z_{(b-1)B+1}, \dots, z_{bB}$ and store them in mini-batch $\caZ_b$.
    \STATE \textcolor{gray}{/* continuous greedy */}
    \STATE Let $\bx_{b}^0 \gets \bfzero$.
    \FOR{$s = 0, \delta, 2\delta, \dots, 1-\delta$}
    \STATE $\tau := \textsc{SmoothTau}(\bx_{b}^s, u, \caZ_b)$
    \STATE $\bg_b^s := \textsc{SmoothGrad}(\bx_{b}^s, \tau, u, \caZ_b)$.
    \STATE Find a vertex $\bd_{b+1}^s$ of $K$ that maximizes $\inprod*{\lambda \sum_{b'=1}^b \bg_{b'}^s + \br_{b}^s, \bd}$ for $\bd \in K$, where $\br_b^s \sim \caD_\mathrm{FPL}$. \COMMENT{FPL}
    \STATE $\bx_{b}^{s + \delta} \gets \bx_{b}^s + \delta \bd_{b}^s$
    \ENDFOR
    \STATE Let $\bx_{b} = \bx_{b}^1$.\label{line:end-FW}
    \ENDFOR
    \RETURN $\bx_{b'}$ for $b'$ chosen from $\{1, 2, \dots, T/B\}$ uniformly at random.
\end{algorithmic}
\end{algorithm}

\begin{algorithm}[t]
    \caption{\textsc{SmoothGrad}$(\bx, \tau, u, \caZ)$}\label{alg:smooth-grad}
\begin{algorithmic}[1]
    \REQUIRE{$\bx$, $\tau$, $u$, and mini-batch $\caZ$}
    \STATE $I_z(\tau) := \max\{ \min\{ \frac{F(\bx; z) - \tau}{u}, 1 \}$ for $z \in \caZ$.
    \RETURN $\sum_{z \in \caZ} I_z(\tau)\nabla_\bx F(\bx, z)$
\end{algorithmic}
\end{algorithm}

\begin{algorithm}[t]
    \caption{\textsc{SmoothTau}$(\bx, u, \caZ)$}\label{alg:smooth-tau}
\begin{algorithmic}[1]
    \REQUIRE{$\bx$, $u$, and mini-batch $\caZ$}
    \STATE $\caB := \{F(\bx; z) \mid z \in \caZ \} \cup \{F(\bx; z) + u \mid z \in \caZ \}$
    \STATE Sort $\caB$ in ascending order, obtaining $\caB = \{b_1, \dots, b_{\abs{\caB}} \}$.
    \STATE $i^* := \min\{i: \sum_{z \in \caZ} I_z(b_i) < \alpha \abs{\caZ}\}$
    \STATE $A := \{z \in \caZ : b_{i^* - 1} < F(\bx; z) < b_{i^*} \}$
    \STATE $C := \{z \in \caZ : F(\bx; z) \leq b_{i^* - 1} \}$
    \RETURN the solution $\tau$ of the linear equation
    \[
        \sum_{z \in \caZ} \frac{F(\bx; z) - \tau}{u} + \abs{C} = \alpha\abs{\caZ}.
    \]
\end{algorithmic}
\end{algorithm}
\subsection{Convergence rate via regret bounds}
Let us consider the convergence rate of \textsc{StochasticRASCAL}.
The main challenge of the analysis is how to set the parameters used in the algorithm, i.e., learning rates $\lambda$, step size $\delta$, perturbation distribution $\caD_\mathrm{FPL}$, smoothing parameter $u$, and mini-batch size $B$, to achieve the desired $O(T^{-1/4})$ convergence rate.

To this end, using tools from \emph{online convex optimization}, we prove an approximate regret bound for a variant of \textsc{StochasticRASCAL} for adversarial online submodular CVaR learning (see Introduction for the definition).

\begin{theorem}[informal]\label{thm:Online-RASCAL-informal}
    There exists an efficient online algorithm for adversarial online CVaR learning with
    \begin{align*}
    \left(1-\frac{1}{e}\right)\sum_{t=1}^T H_t(\bx^*, \tau^*) - \E\left[\sum_{t=1}^T H_t(\bx_t, \tau_t) \right] &= O(T^{3/4})
    \end{align*}
    for an arbitrary $\bx^* \in K$ and $\tau^* \in [0,1]$, where the big-O notation hides a factor polynomial in $\alpha, \beta, D, G$, and $n$.
\end{theorem}

We then show that the above regret bound can be used to show a convergence rate of \textsc{StochasticRASCAL}.
The technical detail of the adversarial setting and the proof of the following theorem is deferred to Appendix.

\begin{theorem}\label{thm:DR-stochastic}
    Under Assumption~\ref{asmp:DR}, \textsc{StochasticRASCAL} outputs $\bx \in K$ such that for any $\bx^* \in K$,
    \begin{align*}
        &\E[\CVaR_\caD(\bx)] \\
        &\geq \left(1-\frac{1}{e}\right)\CVaR_\caD(\bx^*) -
        O\left(\frac{\sqrt{C_\alpha GD}n^{1/8}}{\sqrt\alpha} \right)  T^{-1/4},
    \end{align*}
    where we set
    $B = \frac{\alpha C_\alpha \sqrt{T}}{DGn^{1/4}}$,
    $\delta = \frac{\alpha^{2}}{D^{2} \left((1+\alpha)G L \sqrt{T} + \alpha \beta T^{1/4} \right)}$,
    $\lambda = \frac{\alpha D n^{1/4}\sqrt{B/T}}{G}$,
    $u = \frac{T^{-1/4}}{(1+1/\alpha)}$,
    and
    $\caD_\mathrm{FPL} = \mathrm{Unif}({[0,1]}^n)$,
    and $C_\alpha := \max\{1, \frac{1}{\alpha} - 1\}$.
    Further, if $K$ is an integral polytope contained in $\{\bx \in {[0,1]}^n : \sum_i x_i = k\}$, then
    \begin{align*}
        &\E[\CVaR_\caD(\bx)] \\
        &\geq \left(1-\frac{1}{e}\right)\CVaR_\caD(\bx^*) -
        O\left( \frac{\sqrt{C_\alpha G_\infty} k^{3/4} \log^{1/4} n}{\sqrt\alpha} \right) T^{-1/4}
    \end{align*}
    for
    $B = \frac{\sqrt{2} C_{\alpha} \sqrt{T} \alpha}{2 G_\infty k^{3/2} \sqrt{\log{\left(n \right)}}}$,
    $\delta = \frac{\alpha^{2}}{D^{2} \left((1+\alpha)G L \sqrt{T} + \alpha \beta T^{1/4} \right)}$,
    $\lambda = \sqrt{\frac{B}{Tk}}$,
    $u = \frac{T^{-1/4}}{(1+1/\alpha)}$,
    and
    $\caD_\mathrm{FPL} = N(\bfzero, I)$.
\end{theorem}

To achieve $\E[\CVaR_\caD(\bx)] \geq  \left(1-\frac{1}{e}\right)\CVaR_\caD(\bx^*) - \eps$ for a desired error $\eps > 0$, \textsc{StochasticRASCAL} requires $O(\frac{D^2G^2\sqrt{n}}{\eps^4})$ samples and $O(\frac{DGn^{1.25}}{\eps^2})$ space, whereas \textsc{RASCAL}~\citep{Wilder2018} requires $O(\frac{n}{\eps^2})$ samples and $O(\frac{n^2}{\eps^2})$ space.
Our algorithm runs in a smaller space when the parameters are of moderate size.
For example, if $\beta = D = O(1)$ and $L = G = o(n^{1/4})$, the space complexity of \textsc{StochasticRASCAL} is better than that of \textsc{RASCAL}.

\section{CVaR Maximization of Discrete Submodular Functions}\label{sec:matroid}
We now present our online algorithm for a monotone submodular set function and a matroid constraint.
Let $f_t : 2^V \to [0,1]$ be a monotone submodular function corresponding to the $t$-th sample and $F_t$ be its multilinear extension for $t = 1, \dots, T$.

The basic idea is to run \textsc{StochasticRASCAL} on the multilinear extensions $F_t$ and the matroid polytope $K$.
However, we must address several technical obstacles.
First, we must compare the output portfolio with the optimal \emph{portfolio};
the error bound in the previous sections compared it with the optimal \emph{solution}.
To this end, we make multiple copies of variables so that we can approximate an optimal portfolio by a uniform distribution over a multiset of feasible solutions.
More precisely, we define a continuous DR-submodular function $\bar{F}_t : K^r \to [0,1]$ by
\[
    \bar{F}_t(\bx^1, \dots, \bx^r)
    = \frac{1}{r}\sum_{i=1}^r F_t(\bx^i)
\]
for some sufficiently large $r$.
Then, we feed $\bar{F}_1, \dots, \bar{F}_T$ to \textsc{StochasticRASCAL}.
Suppose that we obtain $(\bx_b^1, \dots, \bx_b^r)$ at Line~\ref{line:end-FW} for each mini-batch $b=1, \dots, T/B$.
Abusing the notation, let us denote $(\bx_t^1, \dots, \bx_t^r) := (\bx_b^1, \dots, \bx_b^r)$ when the $t$-th sample is in the $b$-th mini-batch.

Next, we need to convert $\bx_t^1, \dots, \bx_t^r$ to feasible sets without significantly deteriorating the values of the multilinear extensions.
To this end, we independently apply \emph{randomized swap rounding}~\citep{Chekuri2010} $q$ times to each $\bx^i$ to obtain feasible sets $X_t^{i, 1}, \dots, X_t^{i, q}$.
Note that randomized swap rounding is oblivious rounding and independent from $F_t$.
We can show that $\frac{1}{q}\sum_{j=1}^q f_t(X_t^{i,j})$ is close to $F_t(\bx_t^i)$ by using a concentration inequality.
Finally, after $T$ rounds, we return the uniform portfolio over all $X_t^{i,q}$.
The pseudocode is given in Algorithm~\ref{alg:matroid}.
Carefully choosing $r$ and $q$, we obtain the following theorem.

\begin{algorithm}
    \caption{Online algorithm for maximizing a monotone submodular set function subject to a matroid constraint.}\label{alg:matroid}
    \begin{algorithmic}[1]
        \STATE Run \textsc{StochasticRASCAL} for $\bar{F_1}, \dots, \bar{F_T}$ and the matroid polytope $K$ and let $(\bx_t^1, \dots, \bx_t^r)$ be the temporary solution at Line~\ref{line:end-FW} in \textsc{StochasticRASCAL} for $t=1,\dots,T$.
        \STATE $X_t^{i,j} \gets \textsc{RandomizedSwapRounding}(\bx_t^i)$ for $t=1,\dots, T$, $i = 1, \dots, r$, and $j = 1,\dots, q$.
        \RETURN Uniform portfolio $\bar{\caP}$ over all $X_t^{i,j}$.
    \end{algorithmic}
\end{algorithm}

\begin{theorem}\label{thm:matroid}
    Algorithm~\ref{alg:matroid} achieves
    \begin{align*}
        &\E [\CVaR(\bar\caP)] \\
        &\geq \left(1- \frac{1}{e} \right)\CVaR_\caD(\caP^*) -  O(k^{3/4}\log^{1/4}(n) T^{-1/10})
    \end{align*}
    for arbitrary portfolio $\caP^*$,
    where we set $r=\tilde{O}(T^{1/5})$ and $q = \tilde{O}(T^{3/4})$
    and the expectation is taken over $z_1, \dots, z_T$ and the randomness of the algorithm.
\end{theorem}

\section{Experiments}\label{sec:experiments}

\begin{figure*}[t!]
\centering
\subfloat[NetScience]{\includegraphics[width=.31\hsize]{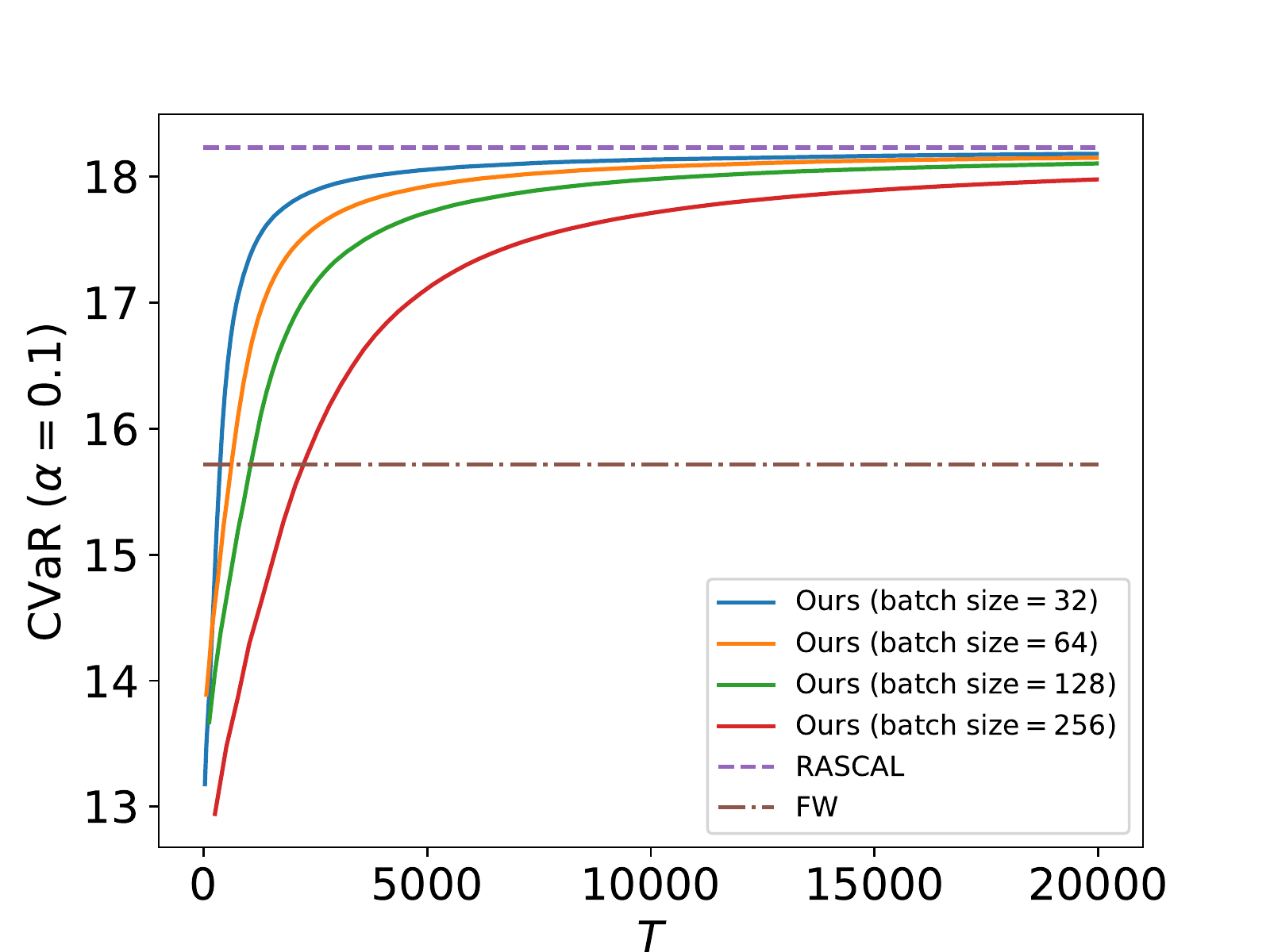}}
\subfloat[EuroRoad]{\includegraphics[width=.31\hsize]{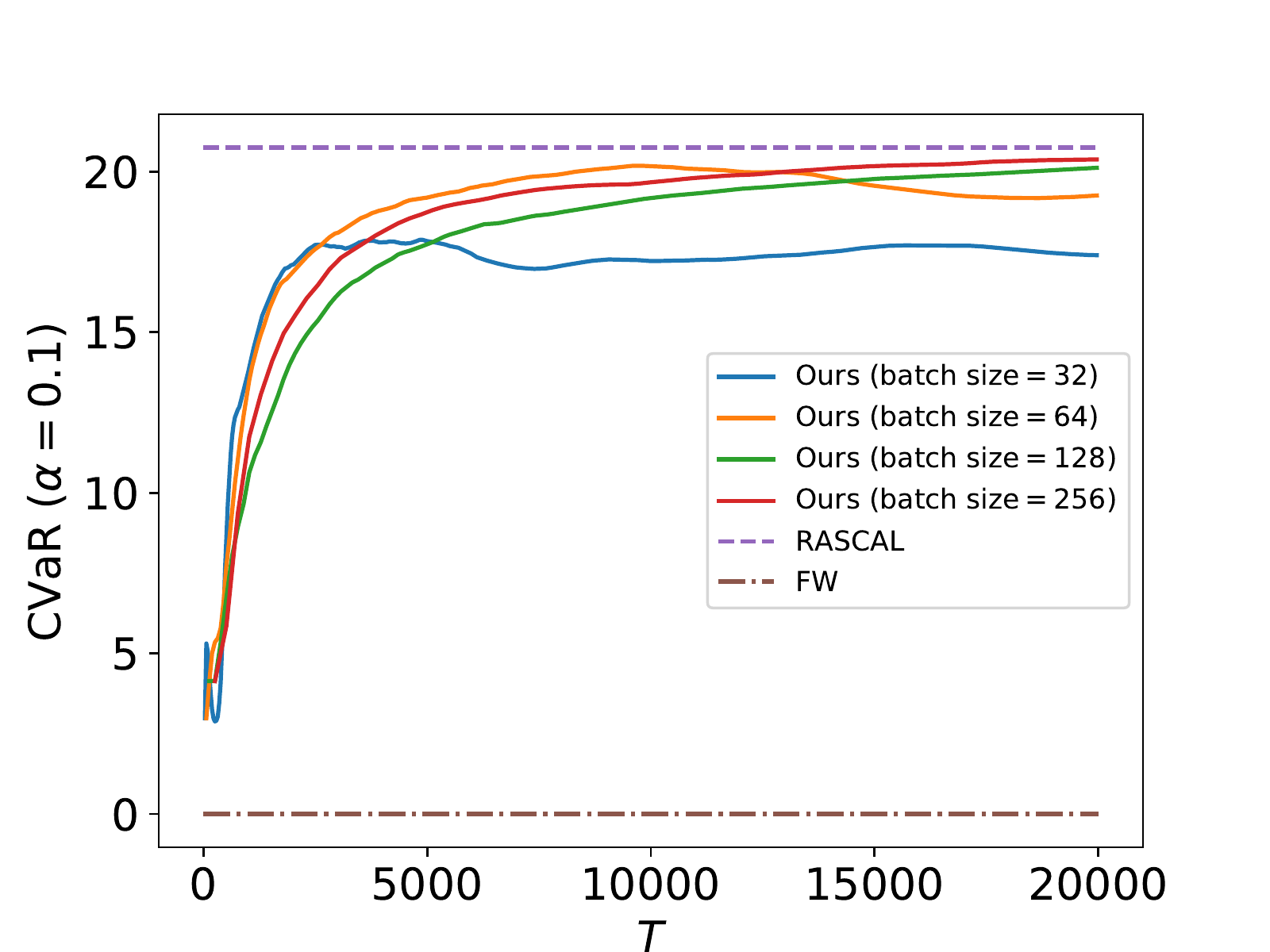}}
\subfloat[BWSN]{\includegraphics[width=.31\hsize]{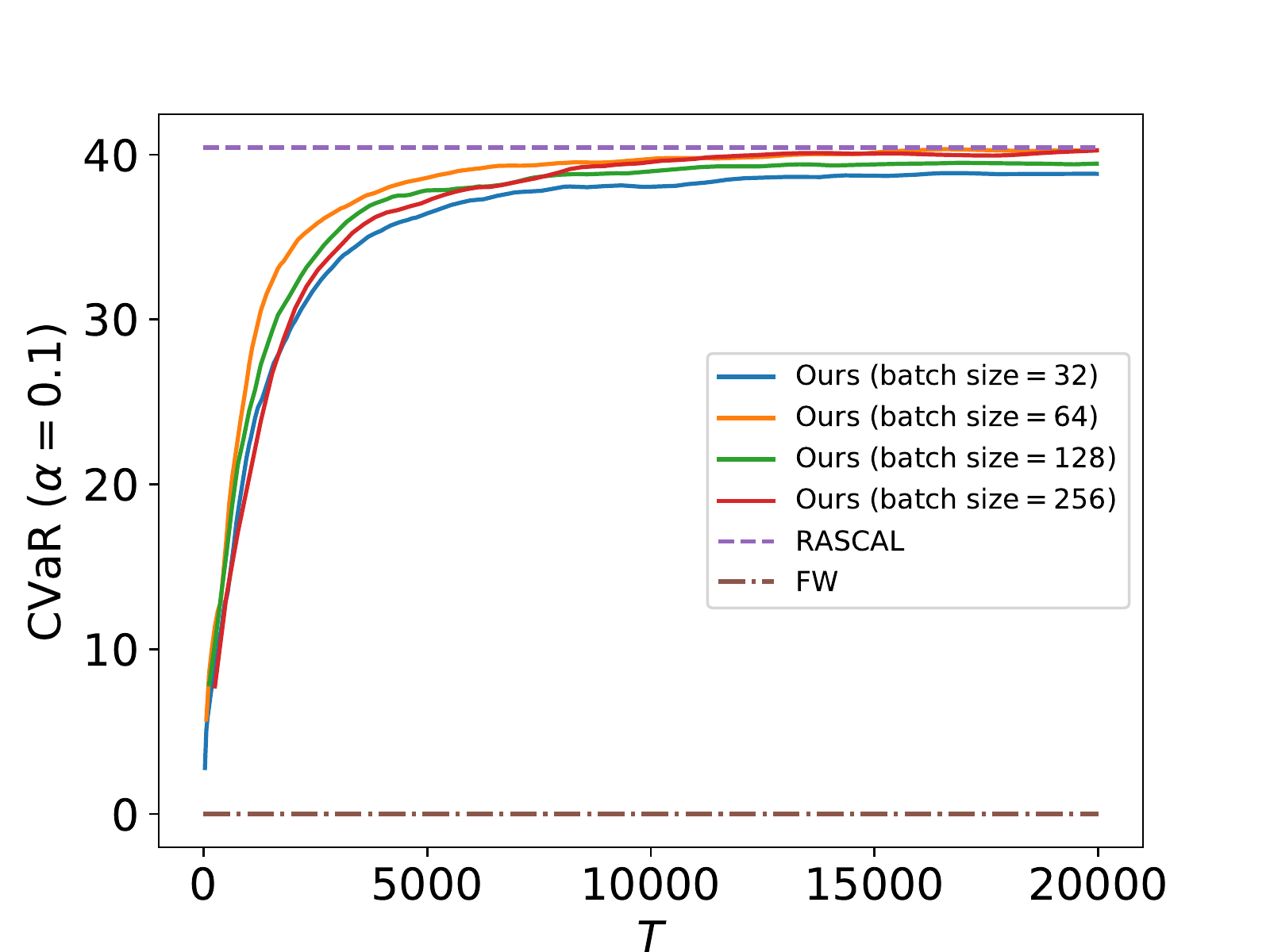}}
\caption{CVaR and the number of samples $T$}\label{fig:T}
\end{figure*}
\begin{figure*}[t!]
\centering
\subfloat[NetScience]{\includegraphics[width=.31\hsize]{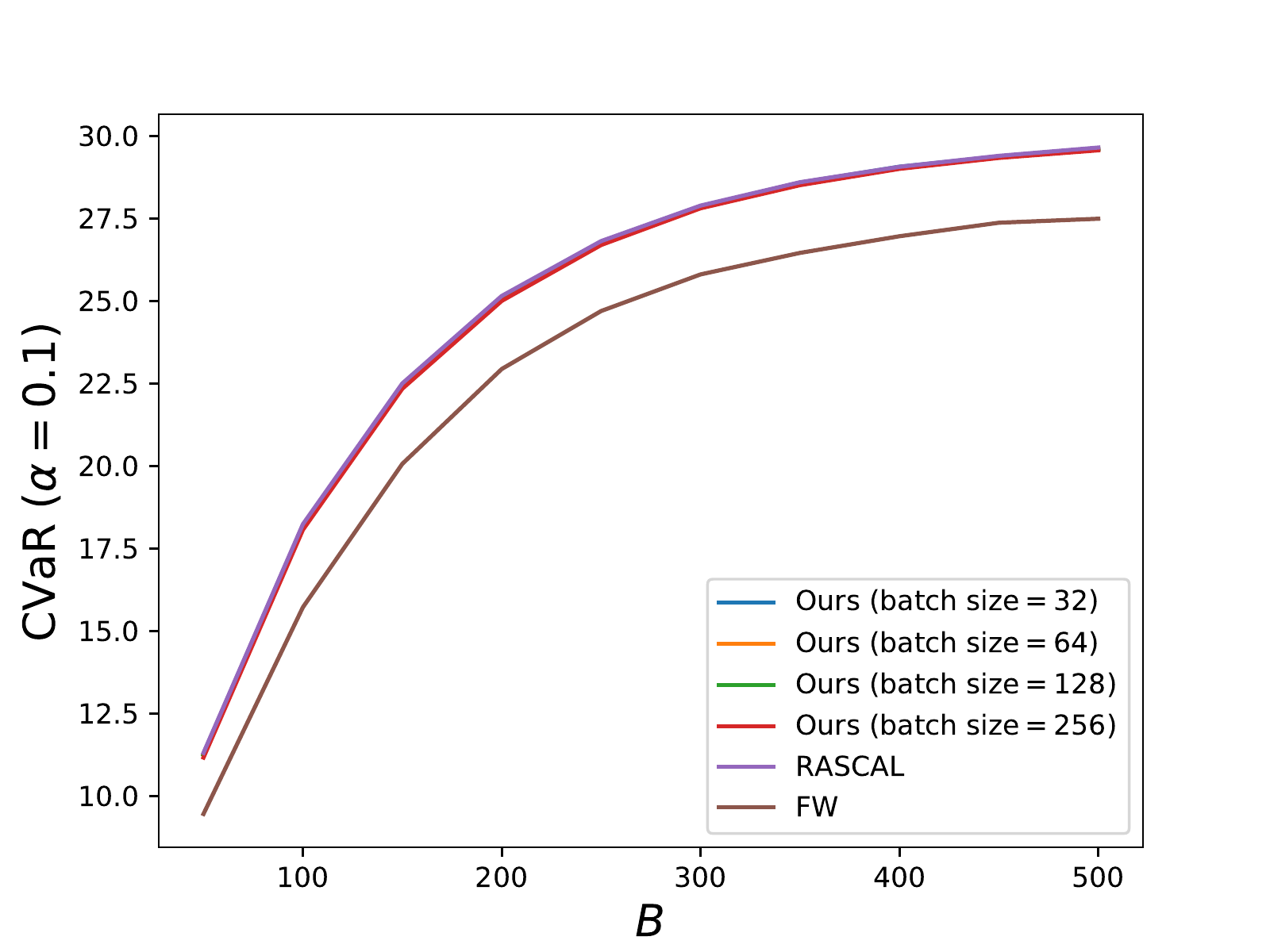}}
\subfloat[EuroRoad]{\includegraphics[width=.31\hsize]{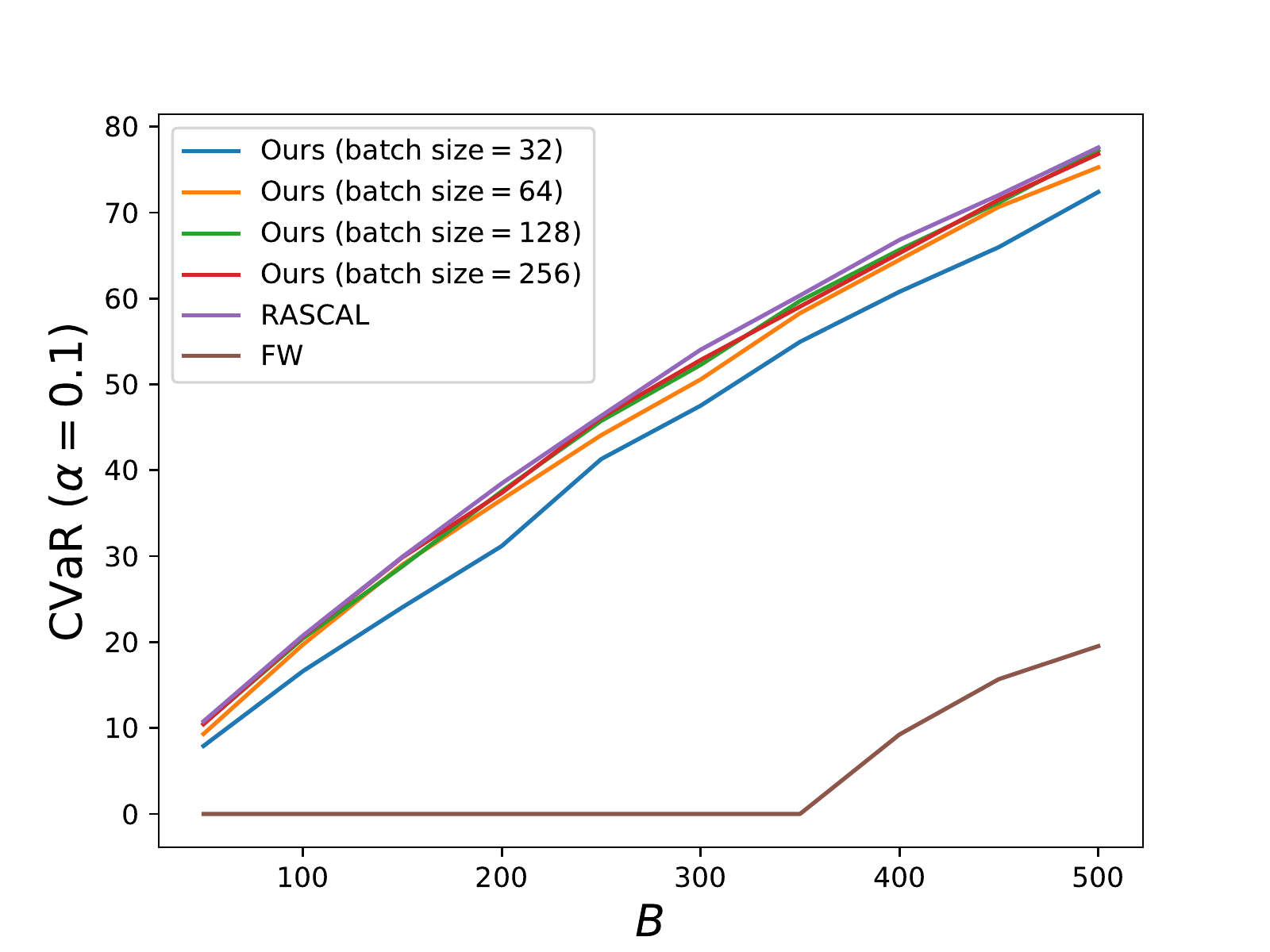}}
\subfloat[BWSN]{\includegraphics[width=.31\hsize]{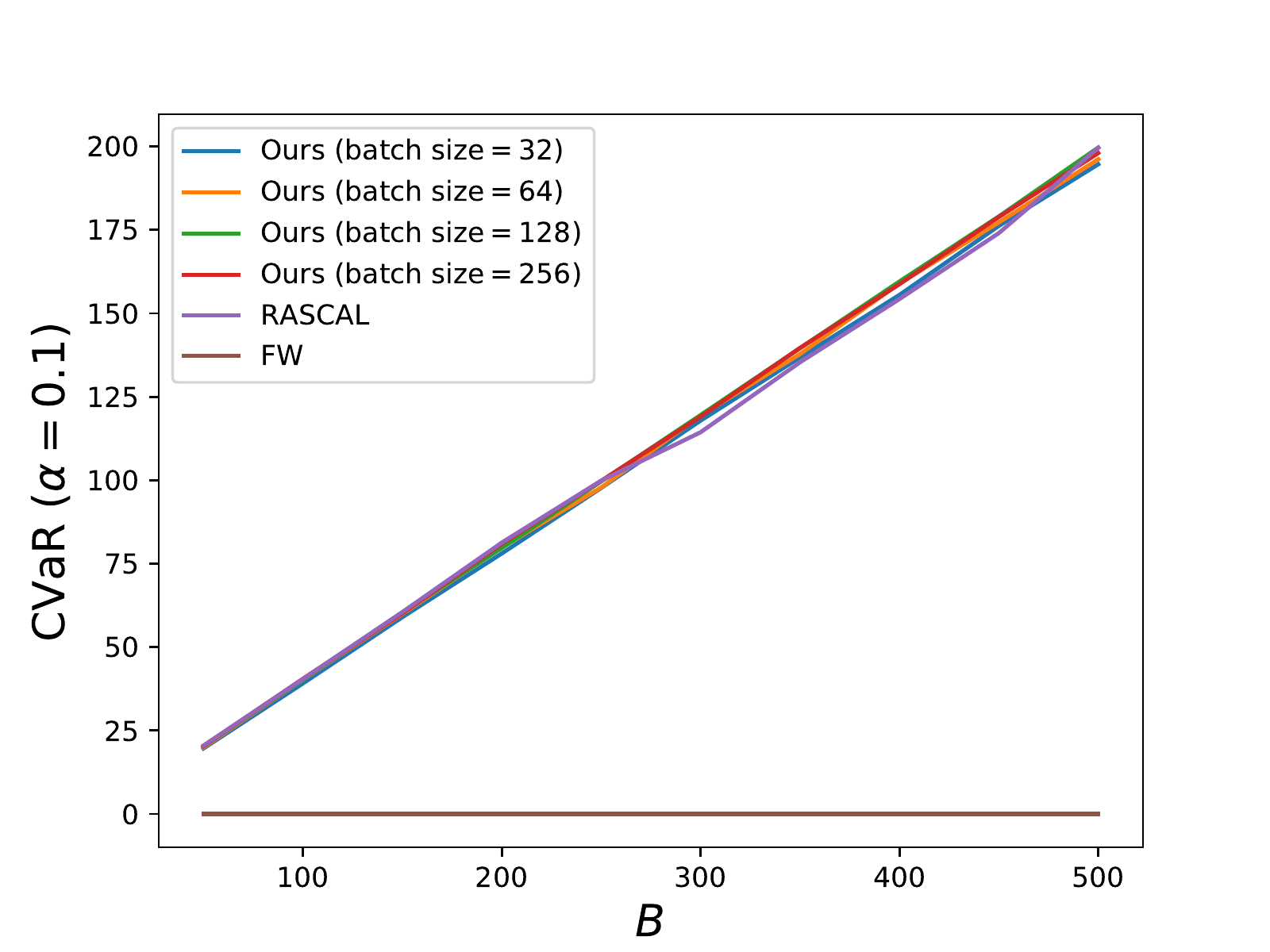}}
\caption{CVaR and budget $B$}\label{fig:B}
\end{figure*}

In this section, we show our experimental results.
In all the experiments, the parameter $\alpha$ of CVaR was set to $0.1$.
The experiments were conducted on a Linux server with
Intel Xeon Gold 6242 (2.8GHz) and 384GB of main memory.

\paragraph{Problem Description.}
We conducted experiments on the sensor resource allocation problem, in which the goal is to rapidly detect a contagion spreading through a network using a limited budget~\citep{Bian2017,leskovec2007cost,soma2015generalization}.

Here, we follow the configuration of the experiments conducted in~\cite{Wilder2018}.
Let $G = (V, E)$ be a graph on $n$ vertices.
A contagion starts at a random vertex and spreads over time according to some specific stochastic process.
Let $z_v$ be the time at which the contagion reaches $v$, and let $z_{\max} = \max_{v \in V : z_v < \infty}z_v$.
If $z_v = \infty$ for some vertex $v \in V$, that is, the contagion does not reach $v$, we reassign $z_v = z_{\max}$, as described in~\cite{Wilder2018}.

The decision maker has a budget $B$ (e.g., energy) to spend on sensing resources.
Let $x_v$ represent the amount of energy allocated to the sensor at a vertex $v$.
When contagion reaches $v$ at time $z_v$, the sensor detects it with probability $1-{(1-p)}^{x_v}$, where $p \in [0,1]$ is the probability that detects the contagion per unit of energy.
The objective $F$ on vectors $\bx = {(x_v)}_{v \in V}$ and $\bz = {(z_v)}_{v \in V}$ is the expected amount of detection time that is saved by the sensor placements:
\[
  F(\bx; \bz) = z_{\max} - \sum_{i=1}^n z_{v_i}(1-{(1-p)}^{x_{v_i}})\prod_{j<i}{(1-p)}^{x_{v_j}},
\]
where the vertices are ordered so that $z_{v_1} \leq z_{v_2} \leq \cdots \leq z_{v_n}$.
It is known that the function $F$ is DR-submodular~\citep{Bian2017}.

\paragraph{Datasets.}
We consider two sensing models and generated three datasets.
In all of them, the source vertex $s$ is chosen uniformly at random.

The first model is the continuous time independent cascade model (CTIC).
In this model, each edge $uv \in E$ has propagation time $\rho_{uv}$ drawn from an exponential distribution with mean $\lambda$.
The contagion starts at the source vertex $s$, i.e., $z_s = 0$, and we iteratively set $z_v = \min_{u\in N(v)} z_u + \rho_{uv}$, where $N(v)$ is the set of neighbors of $v$.
Note that $z_v$ is the first time that the contagion reaches $v$ from its neighbor.
We generated datasets using two real-world networks\footnote{\url{http://konect.cc}}: NetScience, a collaboration network of 1,461 network scientists, and EuroRoad, a network of 1,174 European cities and the roads between them.
For both networks, we set $\lambda = 5$ and $p = 0.01$, and we generated 1,000 scenarios.

The second model, known as the Battle of Water Sensor Networks (BWSN)~\citep{ostfeld2008battle}, involves contamination detection in a water network.
BWSM simulates the spread of contamination through a 126-vertex water network consisting of junctions, tanks, pumps, and the links between them, and the $z_v$ values are provided by a simulator.
We set $p = 0.001$ and generated 1,000 scenarios.

\paragraph{Methods.}
We compared our method against two offline algorithms, \textsc{RASCAL}~\citep{Wilder2018} and the Frank--Wolfe (FW) algorithm~\citep{Bian2017}.
We note that the latter algorithm is designed to maximize the expectation of a DR-submodular function instead of its CVaR.
We run those offline methods on the generated 1,000 scenarios for each dataset.
As our method is an online algorithm, we run our method on 20,000 samples in an online manner, where each sample was uniformly drawn from the set of generated scenarios.

\paragraph{Results.}
Figure~\ref{fig:T} shows how the CVaR changes as $T$ increases.
For each dataset, as long as the batch size is not excessively small, the CVaR attained by our method approaches to that attained by \textsc{RASCAL}\@.
FW algorithm showed significantly lower performance because it is not designed to maximize CVaR.

Figure~\ref{fig:B} shows how the CVaR changes as the budget $B$ increases.
For our method, we plotted the CVaR after processing $10,000$ samples.
We can again confirm that the CVaR attained by our method is close to that attained by \textsc{RASCAL}\@.

\section{Conclusion}
We devised \textsc{StochasticRASCAL} for maximizing CVaR of a monotone stochastic submodular function.
We showed that \textsc{StochasticRASCAL} finds a ($1-1/e$)-approximate solution with a convergence rate of $O(T^{-1/4})$ for monotone continuous DR-submodular functions.
We extended it to portfolio optimization for monotone submodular set functions under a matroid constraint.
Experiments using CTIC and BWSN datasets demonstrated that our algorithm can rapidly achieve CVaRs that are comparable to RASCAL.

\section*{Acknowledgments}
T.S.~is supported by JST, ERATO, Grant Number JPMJER1903, Japan.
Y.Y.~is supported in part by JSPS KAKENHI Grant Number 18H05291 and 20H05965.

\newpage
\bibliographystyle{named}
\bibliography{abbrv.bib}

\onecolumn
\appendix
\section{Adversarial Setting}\label{sec:adversarial}
In this section, we present online algorithm for CVaR maximization in an adversarial environment.

\subsection{Preliminaries on Online Convex Optimization}
We use the framework called \emph{online convex optimization (OCO)} extensively, and will briefly explain it below.
For details, the reader is referred to a monograph~\citep{Hazan2016OCO}.
In OCO, the learner is given a compact convex set $K \subseteq \R^n$.
For each round $t=1,\dots, T$, the learner must select $\bx_t \in K$ and then the adversary reveals a concave reward function\footnote{Although OCO is usually formulated as convex minimization, we state OCO in the form of concave maximization for later use. Note that we can convert minimization to maximization by negating the objective function.} $f_t : K \to [0,1]$ to the learner.
The goal of the learner is to minimize the (1-)regret
\begin{align*}
    \regret(T) = \max_{\bx^* \in K} \sum_{t=1}^T f_T(\bx^*) - \sum_{t=1}^T f_t(\bx_t).
\end{align*}
An important subclass of OCO is \emph{online linear optimization (OLO)}, in which the objective functions are linear.

We use the following OCO algorithms.
Let $\eta > 0$ be a learning rate.

\paragraph{Online Gradient Descent (OGD)}
\[
    \bx_{t+1} = \proj_K (\bx_t + \eta\nabla f_t(\bx_t))
\]

\paragraph{Follow the Perturbed Leader (FPL)}
\[
    \bx_{t+1} = \argmax_{\bx \in K} \inprod*{\eta\sum_{t'=1}^t \nabla f_{t'}(\bx_{t'}) + \br_t, \bx}
\]
where $\br_t$ is a perturbation term drawn from some distribution $\caD_\mathrm{FPL}$.

\begin{lemma}[\citet{Hazan2016OCO}]
    If $\norm{\nabla f_t} \leq G$ and the diameter of $K$ is $D$, then the OGD sequence $\bx_t$ satisfies
    \[
        \regret(T) \leq \eta TG^2 + \frac{D^2}{\eta}.
    \]
\end{lemma}

\begin{lemma}[\citet{Hazan2016OCO,Cohen2015}]
    If $f_t(\bx) = \inprod{\bg_t,\bx}$, $\norm{\bg_t} \leq G$, and the diameter of $K$ is $D$, then the FPL sequence $\bx_t$ with $\caD_\mathrm{FPL} = \mathrm{Unif}([0,1]^n)$ satisfies
    \[
        \E[\regret(T)] \leq \eta TG^2 + \frac{\sqrt{n} D^2}{\eta},
    \]
    where the expectation is taken over the randomness in the algorithm.
    Further, if $\norm{\bg_t}_\infty \leq 1$ and $K$ is an integral polytope contained in $\{\bx \in [0,1]^n : \sum_i x_i = k \}$, then
    \[
        \regret(T) \leq k\sqrt{2\log n} \left( \eta kT + \frac{1}{\eta} \right)
    \]
    for $\caD_\mathrm{FPL} = N(\bfzero, I)$.
\end{lemma}

\subsection{Online Algorithm for Adversarial Online CVaR Learning}
We now present our online algorithm, \textsc{OnlineRASCAL}, for adversarial online submodular CVaR learning.
Let $F_t : \R^n \to [0,1]$ be a monotone continuous DR-submodular function with $F_t(\bfzero) = 0$ for $t=1, \dots, T$.
Let $K \subseteq \R_{\geq 0}^n$ be a down-closed convex set.
\textsc{OnlineRASCAL} maintains $(\bx_t, \tau_t) \in K \times [0,1]$.
Let us consider the $b$-th mini-batch and denote a variable in this mini-batch by $\bx_b$.
Within the $b$-th mini-batch, we play the same $\bx_b$ and use OGD to learn $\argmax_{\tau} \sum_{t=(b-1)B + 1}^{bB} \tilde{H}_t(\bx_b, \tau)$.
At the end of the $b$-th mini-batch, we update $\bx_b$ using online continuous greedy where each inner iteration performs FPL.
The pseudocode can be found in Algorithm~\ref{alg:Online-RASCAL}.

\begin{algorithm}[h]
    \caption{\textsc{OnlineRASCAL}}\label{alg:Online-RASCAL}
\begin{algorithmic}[1]
    \REQUIRE{learning rates $\eta, \lambda > 0$, step size $\delta > 0$, FPL distribution $\caD_\mathrm{FPL}$, and mini-batch size $B > 0$}
    \STATE Initialize $\bx_1 \in K$.
    \FOR{$b=1,\dots,T/B$}
    \STATE \textcolor{gray}{/* learn $\tau$ in mini-batch */}
    \STATE Initialize $\tau_{(b-1)B + 1} \in [0,1]$.
        \FOR{$t = (b-1)B + 1, \dots, bB$}
        \STATE Play $(\bx_b, \tau_t)$ and observe $F_t$.
        \STATE $\gamma_{t} \gets \partial_\tau \tilde{H}_t(\bx_b, \tau_t)$.
        \STATE $\tau_{t+1} \gets \proj_{[0,1]}(\tau_{t} + \eta\gamma_t)$.
        \COMMENT{OGD}
        \ENDFOR
    \STATE Let $\bar{H}_b(\bx) = \frac{1}{B} \max_{\tau}\sum_{t=(b-1)B+1}^{bB} \tilde{H}_t(\bx, \tau)$.
    \STATE \textcolor{gray}{/* continuous greedy */}
    \STATE Let $\bx_{b+1}^0 \gets \bfzero$.
    \FOR{$s = 0, \delta, 2\delta, \dots, 1-\delta$}
    \STATE Compute $\bg_b^s = \nabla_{\bx} \bar{H}_b(\bx_b^s)$ via \textsc{SmoothGrad} and \textsc{SmoothTau}.
    \STATE Take a vertex $\bd_{b+1}^s$ of $K$ that maximizes $\inprod*{\lambda \sum_{b'=1}^{b} {\bg_{b'}^s} + \br_{b}^s, \bd}$ for $\bd \in K$, where $\br_b^s \sim \caD_\mathrm{FPL}$. \COMMENT{FPL}
    \STATE $\bx_{b+1}^{s + \delta} \gets \bx_{b+1}^s + \delta \bd_{b+1}^s$
    \ENDFOR
    \STATE Let $\bx_{b+1} = \bx_{b+1}^1$.
    \ENDFOR
\end{algorithmic}
\end{algorithm}

\begin{theorem}\label{thm:Online-RASCAL-smooth}
    Under Assumption~\ref{asmp:DR}, \textsc{OnlineRASCAL} achieves
    \begin{align*}
        \left(1- \frac{1}{e} \right)\sum_{t=1}^T \tilde{H}_t(\bx^*, \tau^*) - \sum_{t=1}^T \tilde{H}_t(\bx_t, \tau_t)
        &\leq
        B\delta\sum_{s=0}^{1-\delta} \regFPL(T/B)
        + \frac{T}{B}\regOGD(B)
        + \delta \beta_u D^2 T
    \end{align*}
    for arbitrary $\bx^*$ and $\tau^*$,
    where $\regFPL$ is the 1-regret of the $s$-th FPL, $\regOGD$ is the 1-regret of OGD, $\beta_u = \frac{1}{\alpha}(\beta+\frac{LG}{u})$,
    and $D$ is the diameter of $K$.
\end{theorem}
Before diving into the formal proof, we outline the proof.
Recall that the 1-regret of the $s$-th FPL algorithm is
\[
    \regFPL(T/B) = \max_{\bd^* \in K} \sum_{b=1}^{T/B} \inprod{\bg^s_b, \bd^*} - \sum_{b=1}^{T/B} \inprod{\bg^s_b, \bd^s_b}
\]
for each $s$.
Then, using the up-concavity of $\bar{H}_b$, we can prove
\begin{align*}
    &\sum_{b=1}^{T/B} \bar{H}_b(\bx^*) - \bar{H}_b(\bx_b^{s})
    \leq \sum_{b=1}^{T/B} \inprod{\nabla_\bx\bar{H}_b(\bx_b^s), \bd_b^s} + \regFPL(T/B)
\end{align*}
for each $s$.
On the other hand, for each $s$, we have
\begin{align*}
   &\sum_{b=1}^{T/B} \bar{H}_b(\bx_b^{s+\delta}) - \bar{H}_b(\bx_b^{s})
   \geq \sum_{b=1}^{T/B} \inprod{\nabla \bar{H}_b(\bx_b^{s}) , \delta \bd_b^s} - \delta^2 \beta_u D^2 \frac{T}{B},
\end{align*}
where $\beta_u$ is the smoothness parameter of $\bar{H}_b$.
By combining these two inequalities, we can show
\begin{align*}
    &\left(1- \frac{1}{e} \right)\sum_{b=1}^{T/B} \bar{H}_b(\bx^*) - \sum_{b=1}^{T/B} \bar{H}_b(\bx_b)
    \leq \delta \sum_{s=0}^{1-\delta} \regFPL(T/B) + \delta \beta_u D^2 \frac{T}{B}.
\end{align*}
via the standard analysis of continuous greedy.

The next step is to relate $\bar{H}_b$ with the regret in terms of $\tilde{H}_t$.
We have
\begin{align*}
    \sum_{t = (b-1)B + 1}^{bB} \tilde{H}_t(\bx_b, \tau_t)
    &\geq \max_{\tau} \sum_{t = (b-1)B + 1}^{bB} \tilde{H}_t(\bx_b, \tau) - \regOGD(B) \\
    &= B \cdot \bar{H}_b(\bx_b) - \regOGD(B),
\end{align*}
by the definition of the 1-regret of OGD\@.
Combining these two bounds, we can prove Theorem~\ref{thm:Online-RASCAL-smooth}.

\begin{proof}[Proof of Theorem~\ref{thm:Online-RASCAL-smooth}]
First, we obtain
\begin{align*}
    \sum_{b=1}^{T/B} \bar{H}_b(\bx^*) - \bar{H}_b(\bx_b^{s})
    &\leq \sum_{b=1}^{T/B} \bar{H}_b(\bx^* \vee \bx_b^s) - \bar{H}_b(\bx_b^{s}) \tag{monotonicty} \\
    &\leq \sum_{b=1}^{T/B} \inprod{\nabla_\bx\bar{H}_b(\bx_b^s), \bx^* \vee \bx_b^s - \bx_b^{s}} \tag{up-concavity}\\
    &\leq \sum_{b=1}^{T/B} \inprod{\nabla_\bx\bar{H}_b(\bx_b^s), \bx^*} \tag{$\bx^* \geq \bx^*\vee\bx_b^s - \bx_b^s$ and $\nabla_\bx\bar{H}_b(\bx_b^s)\geq\bfzero$} \\
    &\leq \sum_{b=1}^{T/B} \inprod{\nabla_\bx\bar{H}_b(\bx_b^s), \bd_b^s} + \regFPL(T/B). \tag{definition of 1-regret}
\end{align*}
On the other hand, we have
\begin{align*}
   \sum_{b=1}^{T/B} \bar{H}_b(\bx_b^{s+\delta}) - \bar{H}_b(\bx_b^{s})
   &\geq \sum_{b=1}^{T/B} \inprod{\nabla \bar{H}_b(\bx_b^{s+\delta}) , \delta \bd_b^s} \tag{up-concavity} \\
   &\geq \sum_{b=1}^{T/B} \left[\inprod{\nabla \bar{H}_b(\bx_b^{s}) , \delta \bd_b^s} - \delta^2 \beta_u \norm{\bd_b^s}^2 \right] \tag{Lipschitz continuouity of $\nabla \bar{H}_b(\cdot)$} \\
   &\geq \sum_{b=1}^{T/B} \inprod{\nabla \bar{H}_b(\bx_b^{s}) , \delta \bd_b^s} - \delta^2 \beta_u D^2 \frac{T}{B}, \tag{$\norm{\bd_b^s} \leq D$}
\end{align*}
Hence, we obtain
\begin{align*}
   &\sum_{b=1}^{T/B} \bar{H}_b(\bx_b^{s+\delta}) - \bar{H}_b(\bx_b^{s})
   \geq \delta\sum_{b=1}^{T/B} (\bar{H}_b(\bx^*) - \bar{H}_b(\bx_b^{s}) ) - \delta \regFPL (T/B) - \delta^2 \beta_u D^2 \frac{T}{B}
\end{align*}
for $s = 0, \delta, 2\delta, \dots, 1-\delta$.
Then,
\begin{align*}
    \sum_{b=1}^{T/B} (\bar{H}_b(\bx^*) - \bar{H}_b(\bx_b^{s+\delta}) )
    &= \sum_{b=1}^{T/B} (\bar{H}_b(\bx^*) - \bar{H}_b(\bx_b^{s}) ) - \sum_{b=1}^{T/B} (\bar{H}_b(\bx_b^{s+\delta}) - \bar{H}_b(\bx_b^{s})) \\
    &\leq \sum_{b=1}^{T/B} (\bar{H}_b(\bx^*) - \bar{H}_b(\bx_b^{s}) )
    - \delta\sum_{b=1}^{T/B} (\bar{H}_b(\bx^*) - \bar{H}_b(\bx_b^{s}) ) + \delta \regFPL (T/B) + \delta^2 \beta_u D^2 \frac{T}{B} \\
    &= \left( 1 - \delta \right) \sum_{b=1}^{T/B} (\bar{H}_b(\bx^*) - \bar{H}_b(\bx_b^{s}) )
    + \delta \regFPL (T/B) + \delta^2 \beta_u D^2 \frac{T}{B}
\end{align*}
for $s = 0, \delta, 2\delta, \dots, 1-\delta$.
Using this inequality inductively, we have
\begin{align*}
    \sum_{b=1}^{T/B} (\bar{H}_b(\bx^*) - \bar{H}_b(\bx_b^{1}) )
    & \leq \left( 1 - \delta \right)^{1/\delta} \sum_{b=1}^{T/B} (\bar{H}_b(\bx^*) - \bar{H}_b(\bx_b^{0}) )
    + \delta \sum_{s=0}^{1-\delta} \regFPL (T/B) + \delta \beta_u D^2 \frac{T}{B} \\
    & \leq \frac{1}{e} \sum_{b=1}^{T/B} \bar{H}_b(\bx^*)
    + \delta \sum_{s=0}^{1-\delta} \regFPL (T/B) + \delta \beta_u D^2 \frac{T}{B},
\end{align*}
where in the last inequality we used $1+x \leq e^x$ for $x \in \R$ and $\bar{H}_b(\bx_b^0) = \bar{H}_b(\bfzero) = 0$.
Rearranging the terms and using $\bx_b := \bx_b^1$, we obtain
\begin{align}\label{eq:approximate-regret-bar}
    &\left(1-\frac{1}{e} \right)\sum_{b=1}^{T/B} \bar{H}_b(\bx^*) - \sum_{b=1}^{T/B} \bar{H}_b(\bx_b)
    \leq \delta \sum_{s=0}^{1-\delta} \regFPL(T/B) + \delta \beta_u D^2 \frac{T}{B}.
\end{align}
Finally,
\begin{align*}
    \sum_{t=1}^T \tilde{H}_t(\bx_t, \tau_t)
    &= \sum_{b=1}^{T/B} \sum_{t=(b-1)B+1}^{bB} \tilde{H}_t(\bx_b, \tau_t) \tag{$\bx_t$ is a constant in each mini-batch} \\
    &\geq \sum_{b=1}^{T/B} \left[ \max_{\tau}\sum_{t=(b-1)B+1}^{bB} \tilde{H}_t(\bx_b, \tau) -\regOGD(B) \right] \\
    &= B \sum_{b=1}^{T/B} \bar{H}_b(\bx_b) - \frac{T}{B}\regOGD(B) \tag{definition of $\bar{H}_b$}\\
    &\geq B \left[ \left(1-\frac{1}{e}\right)\sum_{b=1}^{T/B} \bar{H}_b(\bx^*) - \delta\sum_{s=0}^{1-\delta} \regFPL(T/B) + \delta \beta_u D^2 \frac{T}{B} \right] - \frac{T}{B}\regOGD(B) \tag{by \eqref{eq:approximate-regret-bar}}\\
    &= \left(1-\frac{1}{e}\right)\sum_{b=1}^{T/B} \max_{\tau} \sum_{t=(b-1)B+1}^{bB} \tilde{H}_t(\bx^*, \tau)
    - B\delta \sum_{s=0}^{1-\delta} \regFPL(T/B) + T\delta \beta_u D^2 - \frac{T}{B}\regOGD(B)\\
    &\geq \left(1-\frac{1}{e}\right)\sum_{b=1}^{T/B} \sum_{t=(b-1)B+1}^{bB} \tilde{H}_t(\bx^*, \tau^*)
    - B\delta \sum_{s=0}^{1-\delta} \regFPL(T/B) + T\delta \beta_u D^2 - \frac{T}{B}\regOGD(B)\\
    &= \left(1-\frac{1}{e}\right)\sum_{t=1}^{T} \tilde{H}_t(\bx^*, \tau^*)
    - B\delta \sum_{s=0}^{1-\delta} \regFPL(T/B) + T\delta \beta_u D^2 - \frac{T}{B}\regOGD(B).
\end{align*}
Thus we have the desired bound in Theorem~\ref{thm:Online-RASCAL-smooth}.
\end{proof}

Finally, we need to convert the approximate regret bound in Theorem~\ref{thm:Online-RASCAL-smooth}, which relates to the smoothed auxiliary functions $\tilde{H}_t$, to the one relating the original auxiliary function $H_t$ for CVaR.
This is achieved by choosing the smoothing parameter $u > 0$ appropriately.
Summarizing the above argument and plugging in the regret bounds of OGD and FPL, we obtain the following.

\begin{theorem}[formal version of Theorem~\ref{thm:Online-RASCAL-informal}]\label{thm:DR}
    Under Assumption~\ref{asmp:DR}, Algorithm~\ref{alg:Online-RASCAL} achieves
    \begin{align*}
        &\left(1-\frac{1}{e}\right)\sum_{t=1}^T H_t(\bx^*, \tau^*) - \sum_{t=1}^T H_t(\bx_t, \tau_t)
        \leq
        B\delta\sum_{s=0}^{1-\delta} \regFPL(T/B) + \frac{T}{B}\regOGD(B)
        + \delta \beta_u D^2 T
        + uT\left(1+\frac{1}{\alpha}\right).
    \end{align*}
    for any $\bx^*$ and $\tau^*$.
    By setting
    $B = \frac{\alpha C_\alpha \sqrt{T}}{DGn^{1/4}}$,
    $\delta = \frac{\alpha^{2}}{D^{2} \left((1+\alpha)G L \sqrt{T} + \alpha \beta T^{1/4} \right)}$,
    $\eta = \frac{1}{C_\alpha\sqrt{B}}$,
    $\lambda = \frac{\alpha D n^{1/4}\sqrt{B/T}}{G}$,
    $u = \frac{T^{-1/4}}{(1+1/\alpha)}$,
    and
    $\caD_\mathrm{FPL} = \mathrm{Unif}({[0,1]}^n)$,
    we obtain
    \begin{align*}
        &\left(1-\frac{1}{e}\right)\sum_{t=1}^T H_t(\bx^*, \tau^*) - \sum_{t=1}^T H_t(\bx_t, \tau_t)
        \leq
        O\left(\frac{\sqrt{C_\alpha GD}n^{1/8}}{\sqrt\alpha} \right)  T^{3/4},
    \end{align*}
    where $C_\alpha := \max\{1, \frac{1}{\alpha} - 1\}$.
    Further, if $K$ is an integral polytope contained in $\{\bx \in {[0,1]}^n : \sum_i x_i = k\}$, then
    \begin{align*}
        &\left(1-\frac{1}{e}\right)\sum_{t=1}^T H_t(\bx^*, \tau^*) - \sum_{t=1}^T H_t(\bx_t, \tau_t)
        \leq
        O\left( \frac{\sqrt{C_\alpha G_\infty} k^{3/4} \log^{1/4} n}{\sqrt\alpha} \right) T^{3/4}
    \end{align*}
    for
    $B = \frac{\sqrt{2} C_{\alpha} \sqrt{T} \alpha}{2 G_\infty k^{3/2} \sqrt{\log{\left(n \right)}}}$,
    $\delta = \frac{\alpha^{2}}{D^{2} \left((1+\alpha)G L \sqrt{T} + \alpha \beta T^{1/4} \right)}$,
    $\eta = \frac{1}{C_\alpha\sqrt{B}}$,
    $\lambda = \sqrt{\frac{B}{Tk}}$,
    $u = \frac{T^{-1/4}}{(1+1/\alpha)}$,
    and
    $\caD_\mathrm{FPL} = N(\bfzero, I)$.
\end{theorem}

\begin{remark}
    One may wonder if we can replace the inner FPL algorithms with \emph{online mirror descent (OMD)}~\citep{Hazan2016OCO} to obtain a (theoretically) better bound.
    A drawback of OMD is that it requires an oracle that computes the Bregman projection on $K$, which is more powerful than the linear optimization oracle used in our algorithm.
    This difference has a significant effect on computational complexity.
    For example, if $K$ is a matroid polytope, then the linear optimization on $K$ can be solved efficiently with the greedy algorithm, whereas the Bregman projection requires minimizing submodular functions multiple times~\citep{Suehiro2012}.
    Therefore, we omit the OMD version in this paper.
\end{remark}

\section{Online to Batch Conversion and Proof of Theorem~\ref{thm:DR-stochastic}}
In this section, we show that an algorithm for adversarial online submodular CVaR learning can be used for CVaR maximization in an i.i.d.\ environment.
As an application, we show that \textsc{StochasticRASCAL} outputs feasible $\bx$ that $\E[\CVaR_\caD(\bx)] \geq (1-1/e)\CVaR_\caD(\bx^*) - O(T^{-1/4})$ for an arbitrary feasible $\bx^*$.
The idea behind the conversion is similar to an \emph{online to batch conversion}~\citep{Cesa2002}

Let $z_1, \dots, z_T$ be i.i.d.\ samples, $F_t = F(\cdot; z_t)$, and $H_t$ be a function defined as
\[
    H_t(\bx, \tau) = \tau - \frac{1}{\alpha}{\left[\tau - F_t(\bx)\right]}_+
    = \tau - \frac{1}{\alpha}{\left[\tau - F(\bx; z_t)\right]}_+,
\]
for $t=1,\dots,T$.
Define a function $H$ by
\[
    H(\bx, \tau) = \tau - \frac{1}{\alpha}\E_{z\sim\caD}{\left[\tau - F(\bx;z)\right]}_+.
\]
Note that $\CVaR_\caD(\caP) = \max_{\tau\in[0,1]} H(\caP,\tau)$.
Further, for fixed $\bx$ and $\tau$, we have $\E_{z_t}H_t(\bx,\tau) = H(\bx,\tau)$.

\begin{lemma}[online to batch]\label{lem:online-to-batch}
    Suppose that an online algorithm outputs $(\bx_t, \tau_t)$ for $F_1, \dots, F_T$.
    Let $\bx$ be sampled uniformly at random from $\bx_1, \dots, \bx_T$.
    Then, for an arbitrary $\bx^* \in K$,
    \begin{align*}
        \E[\CVaR_\caD(\bx)]
        &\geq \left(1-\frac{1}{e}\right)\CVaR_\caD(\bx^*) - \frac{1}{T}\E[\regret_{1-1/e}(T)],
    \end{align*}
    where the expectation is taken over $z_1,\dots,z_T$ and the randomness of the online algorithm.
\end{lemma}

\begin{proof}
    First, we note that
    \begin{align*}
        \E_{z_1,\dots,z_T} H_t(\bx_t, \tau_t)
        &= \E_{z_1,\dots,z_{t-1}} \E_{z_t}[H_t(\bx_t, \tau_t) \mid z_1, \dots, z_{t-1}] \\
        &= \E_{z_1,\dots,z_{t-1}} H(\bx_t, \tau_t),
    \end{align*}
    where in the last equality we used $\bx_t, \tau_t$ is independent to $z_t$ by the definition of online algorithms.
    For an arbitrary $\bx^* \in K, \tau^* \in [0,1]$, we have
    \begin{align*}
        \E[\CVaR_\caD(\bx)]
        &= \frac{1}{T}\sum_{t=1}^T \E_{z_1,\dots,z_T}[\CVaR_\caD(\bx_t)] \\
        &\geq \frac{1}{T}\sum_{t=1}^T \E_{z_1,\dots,z_T}[H(\bx, \tau_t)] \tag{since $\CVaR_\caD(\bx_t) = \max_{\tau \in [0,1]}H(\bx_t, \tau) \geq H(\bx_t, \tau_t)$}\\
        &= \frac{1}{T}\sum_{t=1}^T \E_{z_1,\dots,z_T} [H_t(\bx_t,\tau_t)]  \\
        &\geq \frac{1}{T}\E_{z_1,\dots,z_T} \left(1-\frac{1}{e} \right)\sum_{t=1}^T H_t(\bx^*,\tau^*) -  \frac{1}{T}\E_{z_1,\dots,z_T} \regret_{1-1/e}(T) \\
        &= \left(1-\frac{1}{e} \right)H(\bx^*,\tau^*) -  \frac{1}{T}\E_{z_1,\dots,z_T} \regret_{1-1/e}(T) \tag{since $\E_{z_1,\dots,z_T}H_t(\bx^*,\tau^*) = H(\bx^*,\tau^*)$}.
    \end{align*}
    Since $\tau^*$ is arbitrary, this implies
    \[
        \E[\CVaR_\caD(\bx)] \geq \left(1-\frac{1}{e} \right)\CVaR_\caD(\bx^*) -  \frac{1}{T}\E[\regret_{1-1/e}(T)].
        \]
\end{proof}

Hence, the online algorithm in the previous section immediately yields a $(1-1/e)$-approximation algorithm for CVaR maximization with a convergence rate of $O(T^{-1/4})$.

\subsection{Proof of Theorem~\ref{thm:DR-stochastic}}
It is easy to see that \textsc{StochasticRASCAL} is obtained by \textsc{OnlineRASCAL}.
We note that in \textsc{StochasticRASCAL}, the update of $\tau_t$ is omitted because $\tau_t$ is unnecessary for an i.i.d.~environment.
In other words, only the existence of $\tau_t$ generated by OGD is required.
Hence, by Lemma~\ref{lem:online-to-batch}, \textsc{StochasticRASCAL} has convergence rate $\E[\regret_{1-1/e}(T)]/T = O(T^{-1/4})$, where $\regret_{1-1/e}(T)$ is the $(1-1/e)$-regret of \textsc{OnlineRASCAL}.
The choice of parameters is immediate from Theorem~\ref{thm:DR}.
\section{Proof of Theorem~\ref{thm:matroid}}
We need the several lemma.
Let $\eps, \delta \in (0,1)$ be arbitrary constants.

\begin{lemma}[\citet{Wilder2018}]\label{lem:uniform-approx}
    Let $\caP^*$ be an arbitrary portfolio.
    For $r = \tilde{O}(1/\eps^2)$, there exists a portfolio $\caP$ such that it is a uniform distribution on $r$ sets and $\CVaR_{\caD,\alpha}(\caP) \geq \CVaR_{\caD,\alpha}(\caP^*) - \eps$.
\end{lemma}

\begin{lemma}[\citet{Wilder2018}]\label{lem:rounding}
    Suppose that we have $q = O(\log (1/\delta)/\eps^3)$ samples $X^1,\dots,X^q$ of \textsc{RandomizedSwapRounding} applied to $\bx$.
    Then, for any submodular function $f$ and its multilinear extension $F$, we have
    $\frac{1}{q}\sum_{j=1}^q f(X^j) \geq F(\bx) - \eps$
    with probability at least $1-\delta$.
\end{lemma}

Now we prove Theorem~\ref{thm:matroid}.
For any $\bx^{*1}, \dots, \bx^{*r}$ and $\tau^*$, we have
\[
   \sum_{t=1}^T H_t(\bx_t^1, \dots, \bx_t^r, \tau_t) \geq \left(1-\frac{1}{e} \right) \sum_{t=1}^T H_t(\bx^{*1}, \dots, \bx^{*r}, \tau^*) - \regret_{1-1/e}(T),
\]
where
\[
   H_t(\bx_t^1, \dots, \bx_t^r, \tau_t) = \tau_t - \frac{1}{\alpha}{\left[ \tau_t - \frac{1}{r}\sum_{i=1}^r F_t(\bx_t^r) \right]}_+
\]
and
\[
    \regret_{1-1/e}(T) \lesssim
        T^{3/4} \frac{{(kr)}^{3/4}{((1+1/\alpha)\log (nr))}^{1/4}}{\sqrt{\alpha}}
\]
by Theorem~\ref{thm:DR}.
Note that $G_\infty = 1$ for the multilinear extensions $F_t$ and $K^r$ is contained in $\{\bx : \sum_i x_i = kr\}$.
Applying \textsc{RandomizedSwapRounding} to $\bx_t^i$ $q$ times, we have
\begin{align}\label{eq:guarantee-rounding}
    \frac{1}{q} \sum_{j=1}^q f_t(X_t^{i,j}) \geq F_t(\bx_t^i) - \eps
\end{align}
with probability at least $1-\delta$ by Lemma~\ref{lem:rounding}.
By the union bound, we have \eqref{eq:guarantee-rounding} for all $i$, $t$ with probability at least $1-rT\delta$.
Hence, by setting $\delta = O(1/rT)$, we only suffer a constant regret from the event that \eqref{eq:guarantee-rounding} does not hold.
Hence, we can assume that \eqref{eq:guarantee-rounding} holds for all $i$, $t$.

Abusing the notation, let us define
\[
    H_t(\caP, \tau) := \tau - \frac{1}{\alpha}[\tau - \E_{X\sim \caP} f_t(X) ]_+
\]
for a portfolio $\caP$ and $\tau \in \R$.
Recall that $\caP_t$ is the portfolio of $X_t^{i,j}$ for each $t$.
Then, by \eqref{eq:guarantee-rounding}, we have
\[
    H_t(\caP_t, \tau_t) \geq H_t(\bx^1, \dots, \bx^r, \tau_t) - \eps
\]
for all $t$.
Summing this up for $t=1, \dots, T$, we obtain
\begin{align*}
    \sum_{t=1}^T H_t(\caP_t, \tau_t)
    &\geq \sum_{t=1}^T H_t(\bx^1, \dots, \bx^r, \tau_t) - \eps T \\
    &\geq \left(1-\frac{1}{e} \right) \sum_{t=1}^T H_t(\bx^{*1}, \dots, \bx^{*r}, \tau^*) -\eps T - \regret_{1-1/e}(T).
\end{align*}
By the online-to-batch conversion, we have
\begin{align*}
    \E_{z_1,\dots,z_T}\CVaR_\caD(\bar\caP)
    &\geq \left(1-\frac{1}{e} \right) \CVaR_{\caD}(\caP_\text{unif}^*) -\eps -  \frac{1}{T}\E_{z_1,\dots,z_T}\regret_{1-1/e}(T),
\end{align*}
where $\caP^*_\text{unif}$ is any uniform distribution on $r$ sets.
Setting $\eps = T^{-1/4}$, which in turn yields $q = \tilde{O}(T^{3/4})$, the entire error will be dominated by $\regret_{1-1/e}(T)$.
Finally, we need to balance $r$.
Setting $r = \tilde{O}(T^{1/5})$, we can approximate an optimal portfolio $\caP^*$ by a uniform portfolio $\caP^*_\text{unif}$ such that $\CVaR_\caD(\caP^*) - \CVaR_\caD(\caP^*_\text{unif}) \leq T^{-1/10}$ and $r^{3/4}/T^{1/4} = T^{-1/10}$ by Lemma~\ref{lem:uniform-approx}.

\end{document}